\documentclass[keywords,showpacs,preprintnumbers,superscriptaddress]{revtex4}

\usepackage{amsfonts,amssymb,amsmath,amsthm,epsfig}
\advance\topmargin 2cm

\newtheorem*{Sperner}{Sperner's Theorem}
\newtheorem{theorem}{Theorem}
\newtheorem{observation}{Observation}

\newtheorem{lemma}{Lemma}

\newtheorem{corollary}{Corollary}
\newtheorem{definition}{Definition}

\newcommand{\nn}{\nonumber}
\newcommand{\ba}{\begin{array}}
\newcommand{\ea}{\end{array}}

\usepackage{graphicx}
\usepackage{dcolumn}
\usepackage{bm}

\begin{document}

\title{Hypergraph-theoretic characterizations for
LOCC incomparable ensembles of multipartite CAT states}\thanks{A
part of this work was done when the first author was visiting the
Perimeter Institute for Theoretical Physics, Waterloo, and the third author
was visiting ETH, Zurich.}

\author{Arijit Ghosh}
\email[E-mail:~]{arijitiitkgpster@gmail.com}

\affiliation{Department of Computer Science and Engineering,\\ 
Indian Institute of Technology, Kharagpur 721302, India.}
\author{Sudebkumar Prasant Pal}
\email[E-mail:~]{spp@cse.iitkgp.ernet.in}
\affiliation{Centre for Theoretical Studies and\\
Department of Computer Science and Engineering,\\ 
Indian Institute of Technology, Kharagpur 721302, India.}

\author{Anupam Prakash}
\email[E-mail:~]{anupamprakash1@gmail.com}

\affiliation{Department of Computer Science and Engineering,\\ 
Indian Institute of Technology, Kharagpur 721302, India.}
\author{Virendra Singh Shekhawat} 
\email[E-mail:~]{shekhawat.virendra@gmail.com}

\affiliation{Department of Computer Science and Engineering,\\ 
Indian Institute of Technology, Kharagpur 721302, India.}

\date{\today}

\begin{abstract}

Using graphs and hypergraphs to systematically model collections of
arbitrary subsets 
of parties representing {\it ensembles (or collections)} of
shared multipartite CAT states, we 
study transformations between such {\it ensembles} under
{\it local operations
and classical communication (LOCC)}.
We show using partial entropic criteria, 
that any two such distinct ensembles
represented by {\it $r$-uniform hypergraphs} with the same number of
hyperedges (CAT states), are LOCC incomparable for even integers 
$r\geq 2$, generalizing results in \cite{mscthesis,sin:pal:kum:sri}.
We show that the cardinality
of the largest set of mutually LOCC incomparable 
ensembles represented by $r$-uniform hypergraphs for even $r\geq 2$, 
is exponential in the number of parties.
We also demonstrate LOCC incomparability between two 
ensembles represented by 3-uniform hypergraphs 
where partial entropic criteria do not help in
establishing incomparability. Further we
characterize LOCC comparability of EPR graphs in a model where LOCC 
is restricted to teleportation and edge destruction. 
We show that 
this model is equivalent to one in which LOCC transformations
are carried out through a sequence of operations where each 
operation adds at most one new EPR pair.

\vspace{0.1in}
\end{abstract}


\maketitle

\section{Introduction}

Certain operations like entanglement teleportation and
creation of multipartite entanglement states can 
be done using only classical communication, with the aid of
preshared quantum entanglement between the geographically
separated parties \cite{teleport,nc}. 
Correlations in different problems being solved 
between various subsets of parties in a scalable network may 
be exploited by using multiple preshared
entanglements within those subsets of parties 
for reducing classical communication complexity. See \cite {bdht,bcd,cb}
for problems where such savings are possible. 
The specifications of requisite patterns of entanglement
may change over a period of time in a quantum computation network;
we may require to solve different problems between different
sets of combinations of
parties. In such a scenario, it becomes necessary to
transform one set of entanglement combinations across the network,
into another distinctly different set of entanglements.
The main question is whether such transformations 
from one pattern of multiple preshared entanglements
between parties to another such pattern are possible
using only LOCC ({\it local operations and classical communication}). 
Nielsen \cite{neil99,nc} derived important  
results about conditions for LOCC transformations between 
bipartite states and the partial order between such states.
Linden et al.  \cite{lin:pop:sch:wes:99}, considered reversible 
transformations using local quantum operations and classical communication
for multi-particle environments. 
For multipartite entanglement ensembles of CAT states
shared between various combinations of parties, several 
important characterizations of
LOCC incomparability were derived in \cite{mscthesis,sin:pal:kum:sri} using
the method of {\it bicolored merging}, based on
{\it partial entropic criteria}.    
In this work we further characterize and classify
certain incomparable ensembles of multipartite CAT sates
combinatorially, and study
LOCC transformations between ensembles that are not incomparable.
We study the scope and limitations of partial entropic criteria
in establishing LOCC incomparability between multipartite states.
Before we outline our contribution, we present a few necessary definitions 
and some notation.



We need a few definitions and some notation.
An {\em EPR graph} $G(V,E)$ is a graph
whose vertex set 
$V$ is a set of {\it parties}, and an edge
$\{u,v\}$, where  $u,v\in V$, represents shared
entanglement in the form of an EPR
pair between the parties $u$ and $v$.
An {\it entanglement configuration hypergraph (EC hypergraph)}
$H(S,F)$, has a set
$S$ of $n$ parties and a set $F= \{E_1,
E_2,\cdots, E_m\}$ of $m$ hyperedges,
where $ E_i  \subseteq S; i = 1, 2, \cdots , m$,
and $E_i$ is such that its elements
(parties) share an $|E_i|$-CAT state.
So, an EPR graph or an EC hypergraph represents multipartite states
with multiple entanglements in the form of CAT states, where each CAT state
is represented by an edge or hyperedge, respectively.
If one such multipartite state $|\phi \rangle$ can be transformed into 
another such state $|\psi \rangle$ by LOCC, 
then we denote this transformation as $|\phi \rangle \geq |\psi \rangle$
(or $|\psi \rangle \leq |\phi \rangle$). 
If none of 
$|\phi \rangle \geq |\psi \rangle$
and $|\psi \rangle \geq |\phi \rangle$
hold then we say that the two ensembles or states are {\it LOCC incomparable}.
If one or both of  
$|\phi \rangle \geq |\psi \rangle$
and $|\psi \rangle \geq |\phi \rangle$
hold, then we say that the two ensembles or states are {\it LOCC comparable}.

If there is a path in a graph between every 
pair of vertices then the
graph is called a {\it connected graph}.
{\it EPR trees} are {\it connected} EPR graphs with 
$n$ vertices and exactly $n-1$ edges. 
A {\it spanning tree} is a graph which connects
all vertices without forming cycles. EPR trees are indeed spanning
trees.
There is a unique path between any two vertices
in a spanning tree.
There may be more than one path between a pair
of vertices in an arbitrary graph.


{\it EC hypertrees} are EC hypergraphs with no {\it cycles}; 
hypertrees have at most one vertex common between any two hyperedges.
{\it Connectedness} for hypergraphs is defined as follows.
An alternating sequence of vertices and hyperedges $\{a,E_1,v_1,
E_2,v_2,\cdots ,E_i,v_i,E_{i+1},\cdots ,E_j,b\}$ 
in a hypergraph  $H=(S,F)$  is  called  a  {\it
hyperpath} from a vertex $a\in S$ to a vertex $b\in S$ if
$E_i$  and $E_{i+1}$ have a common vertex $v_i$ in $S$, for  all $1\leq i\leq
j-1$, $a\in E_1$, and $b\in E_j$, where the vertices $v_i$, $1\leq i\leq j-1$,
are distinct, and 
the hyperedges $E_i$, $1\leq i\leq j$ are distinct.
If the start and end vertices of a hyperpath are identical 
then the hyperpath is called a {\it cycle} and the hypergraph is said to be
{\it cyclic}.
If the hypergraph $H$ has a hyperpath between every pair $a,b\in S$,
then $H$ is said to be {\it connected}.

The {\it degree} of a vertex in a (hyper)graph is the 
number of (hyper)edges containing
that vertex, in the (hyper)graph. 
The {\it degree} of a vertex subset in a hypergraph
is the number of hyperedges containing
all vertices of the vertex subset, in the hypergraph. 
We use $E(G)$ to denote the set of (hyper)edges of a (hyper)graph $G$.


%



Partial entropic arguments as applicable to ensembles of multipartite CAT 
states were used in the method of {\it bicolored merging} 
as in \cite{mscthesis,sin:pal:kum:sri} to establish several important 
LOCC incomparability results. We begin this paper by discussing the
equivalence of partial entropic criteria and the method of bicolored merging
in Section \ref{bipar},
as applied to such ensembles as EPR graphs and EC hypergraphs.  
In Appendix A, we elaborate a formal proof of this equivalence.
It was shown in \cite{sin:pal:kum:sri} that
(i) $n-2$ copies of $n$-CAT states shared between $n$ geographically
separate parties cannot be converted
to an {\it EPR tree} shared between the $n$ parties using only $LOCC$, and
(ii) two distinct {\it $r$-uniform EC hypertrees}
shared between $n$ geographically separated parties
are $LOCC$ incomparable.
In this paper we further demonstrate the power of partial entropic 
criteria in Section \ref{comec} by proving the LOCC incomparability 
of any two distinct {\it $r$-uniform EC hypergraphs} having 
the same number of hyperedges, for even integers $r\geq 2$, generalizing
results in \cite{mscthesis,sin:pal:kum:sri} for $r$-uniform EC hypertrees.
The question of incomparability remains open for $r$-uniform 
EC hypergraphs with the same number of hyperedges, for odd values
of $r\geq 3$. We conjecture that incomparability
holds for odd values of $r$ as well. 
In order to establish incomparability results we use 
(i) the {\it inclusion-exclusion
principle}, and (ii) the generalized notion of {\it degree of a vertex subset},
to model partial entropy of {\it reduced or collapsed hypergraphs}.
Using the same technique, we present a significantly simpler proof of 
the LOCC incomparability result in \cite{sin:pal:kum:sri} 
of distinct $r$-uniform EC hypertrees for all integers $r\geq 2$. 
We observe that changing the set of hyperedges but keeping 
the total amount of entanglement (in the sense of the number of CAT states or
hyperedges in an ensemble) fixed,
induces incomparability in these distinct EC hypergraphs. 
Changing the set of hyperedges in this manner, we can generate numerous
mutually LOCC incomparable hypergraphs in 
a natural {\it partial order} called the {\it LOCC partial order},
defined as follows.
A node in this {\it LOCC partial order}
represents an equivalence class of states that are mutually {\it LU 
(locally unitary) 
equivalent}. The directed edge $(v,w)$ exists in the
{\it directed acyclic graph} representing this partial order,
if any state in the
equivalence class of $v$ can be transformed by LOCC to
a state in the equivalence class $w$.
In Section \ref{widthsec}, 
we obtain
the maximum number of mutually LOCC incomparable  
$r$-uniform EC hypergraphs, using {\it Sperner's Theorem}
\cite{sperner} for even $r$. This yields the  
{\it width} of the LOCC partial order, which is exponential in $n$. 
This demonstrates the necessity of quantum communication
for transformations between these numerous incomparable states. 
These results are reported in \cite{btechthesis}.

In this paper we also study incomparability for ensembles where partial entropy 
criteria are not useful in establishing incomparability.
The well known example of 3EPR-2GHZ falls in this category
(see \cite{bennett00}). 
In Section \ref{comec},
we provide another example of a pair of
ensembles represented by $3$-uniform 
hypergraphs having 4 hyperedges each, which cannot be shown to
be incomparable using partial entropic criteria.
The EC hypergraphs $H_1$ and $H_2$ representing these ensembles 
have hyperedge sets $E_1 =\{\{123\},\{156\},\{245\},\{346\}\}$
and $E_{2}=\{\{456\},\{234\},\{136\},\{125\}\}$, respectively.
We establish the incomparability of these hypergraphs
using LU inequivalence of reduced states 
following results in \cite{bennett00}. 
As conjectured above, we believe that any two $r$-uniform hypergraphs 
with equal number of hyperedges are LOCC incomparable for
all integers $r\geq 2$.  

Finally, using combinatorial techniques 
in Section \ref{cutssec}, 
we characterize LOCC comparability for EPR graphs in a model where
LOCC is restricted to the operations of {\it edge destruction} 
and {\it teleportation}. The NP-completeness of the problem 
of deciding LOCC comparability in this restricted model follows 
from results in \cite{disjoint}. 
We also show that restricted LOCC is 
equivalent to a model of LOCC where new edges are added 
one at a time.

\section{Partial entropy, bicolored merging and LOCC incomparability}
\label{bipar} 
Suppose we create a bipartition amongst the $n$
parties in such a way that the partial entropy is different for the
two given states, $|\psi \rangle$ and $|\phi \rangle$, where these
states are represented by EPR graphs or EC hypergraphs. In the case
of EPR graphs, the difference in partial entropy between the two
states is simply the difference between the number of EPR pairs
shared across the partition in the two states. In the case of
multipartite states represented by EC hypergraphs, the
difference in partial entropies is equal to the difference in
the number of multipartite CAT states shared across the partition in
the two states. In both these cases, the state corresponding to the
higher entropy cannot be obtained from that with lower entropy, as
long as only $LOCC$ is used. 

In order to show that a
multipartite state 
$|\psi \rangle$ 
can not be converted to the
multipartite state 
$|\phi \rangle$ 
by $LOCC$, we may
partition the original set of parties into 
(only) two hypothetical {\it merged parties or entities}
say, $A$ and $B$. We may also view this as coloring the parties with two
colors, one for those assigned to $A$ and the other for those to $B$.
Parties in set $A$ are merged into one single party. Similarly,
parties in set $B$ are merged into another single party. Each
hyperedge shared between parties of $A$ and $B$ in 
$|\psi \rangle$ 
(or $|\phi \rangle$) 
is reduced to a
single hypothetical EPR pair between the merged parties $A$ and $B$
resulting in the 
{\it bicolor merged graph (BCM)} say $H_1^{bcm}$ (or $H_2^{bcm}$), as 
defined  
in \cite{sin:pal:kum:sri}.
Then we count the
number of (such hypothetical) EPR-pairs shared across
the merged parties $A$ and $B$ in these two graphs. 
If $H_1^{bcm}$
has a smaller such
count then 
$|\psi \rangle$  
cannot be transformed by LOCC into the other state
$|\phi \rangle$ 
whose BCM
$H_2^{bcm}$
has a larger count. The partitioning into two parts
and the collapsing of all parties into these two parts is referred to
as {\it bicolored merging} in \cite{sin:pal:kum:sri}. See Figures
\ref{figure4} and \ref{figure5} for an illustration. 

\begin{figure}[ht]
\centerline{\epsfig{figure=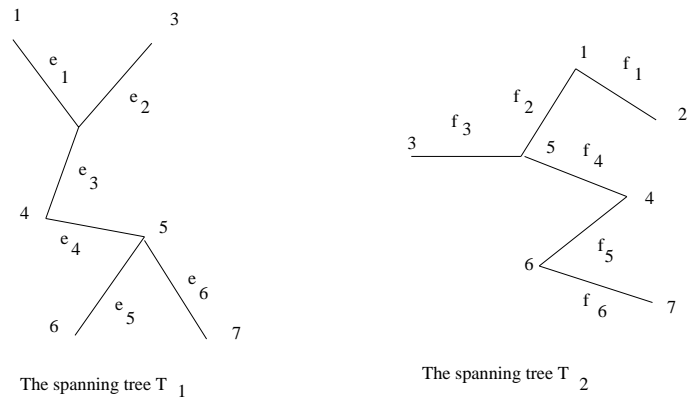,width=90mm}} \caption{Two
LOCC incomparable 7-vertex EPR spanning trees (2-uniform
hypertrees) from \cite{sin:pal:kum:sri}.} \label{figure4}
\end{figure}

\begin{figure}[ht]
\rotatebox{0}{\centerline{\epsfig{figure=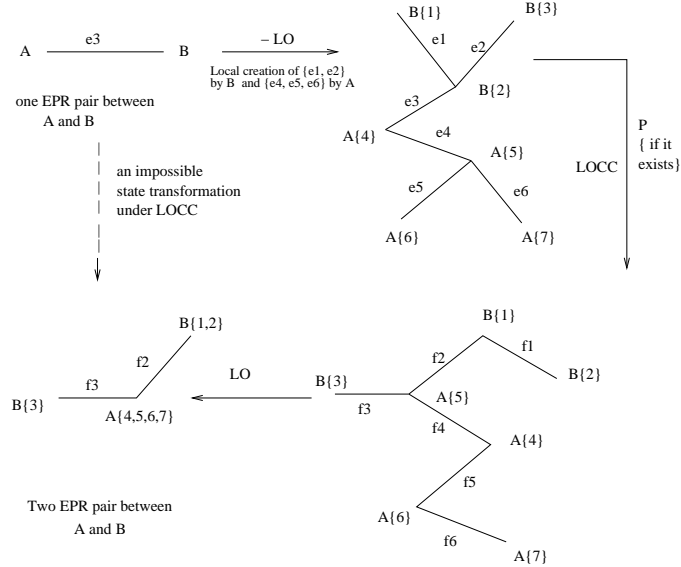,width=90mm}}}
\caption{Illustration of the method of {\it bicolored merging}, 
establishing LOCC incomparability as in \cite{sin:pal:kum:sri}.}
\label{figure5}
\end{figure}

Now we formally state the scope of the technique of {\it bicolored
merging} in establishing LOCC incomparability of EPR graphs and EC
hypergraphs in graph theoretic notation and terminology as follows.
Let $H_1$ and $H_2$ be two {\it EC hypergraphs}
shared between $n$ geographically separated parties such that
a bipartition $(A,B)$ of set of $n$ vertices has strictly smaller
partial entropy for $H_1$. So $H_1$ cannot be transformed to 
$H_2$ by LOCC.
The partition $(A,B)$ of the set of parties may be viewed as
as a $cut$ of the hypergraph, cutting across hyperedges that have at
least one vertex in each of the parts $A$ and $B$. The number of
hyperedges of $H_1$ ($H_2$)
across the cut is called the $capacity$ of the cut in the
respective hypergraph.

\begin{observation}
  If the $capacity$ of a cut $(A,B)$ is strictly smaller in EC hypergraph
  $H_1$ than in EC hypergraph $H_2$, then
  $H_1$ cannot be transformed into $H_2$ by LOCC.
\label{obscut}
\end{observation}

Bicolored merging, or equivalently, the method of partial
entropy may not help in establishing the LOCC incomparability
of certain pairs of states.
The example of the two states viz., 3EPRs and 2GHZs,
shared between three parties as in
\cite{bennett00}, 
is an example where partial entropic methods cannot
help us in establishing their LOCC incomparability.

%
%

\section{Combinatorics of entanglement configuration hypergraphs (EC
hypergraphs)}

\label{comec}

LOCC operations can transform one EC
hypergraph into another.
There are examples of large sets of EC hypergraphs
that are mutually LOCC
incomparable.
One such set is that of labeled $r$-uniform hypertrees
\cite{sin:pal:kum:sri}.
We first establish results for EPR graphs and
then generalize them to certain classes of EC hypergraphs.
The following lemma applies to EC hypergraphs and EPR graphs;
the proof is presented for the general case of EC hypergraphs.
\begin{lemma}
\label{degreefall}
The degree of a vertex $v$ in an EC hypergraph (or in an EPR
graph) cannot increase under $LOCC$ transformations.
\end{lemma}
\begin{proof}
Let $H_1$ be a
{\it EC  hypergraph} which can 
be transformed into another
EC hypergraph $H_2$ by $LOCC$. For a
vertex $v\in H_1$,
define a bipartition of $H_1$ by
placing $v$ in one set
and the remaining vertices in the other.
The number of edges across
the cut defined by this bipartition is
equal to the degree of $v$.
By the contrapositive of Observation \ref{obscut} above, it follows
that the degree of $v$ cannot
increase under LOCC.
\end{proof}

The above lemma is of great importance
as it provides a localized view
of a party and states that its total entanglement
measure with other parties does not go up under $LOCC$ operations.
Using this result we generalize
the incomparability result of EPR trees as in
\cite{sin:pal:kum:sri} to $EPR$ graphs with
the same number of edges.
\begin{theorem}
\label{EPRincom}
Any two distinct labeled
$EPR$ graphs with the same number of vertices
and edges are $LOCC$ incomparable.
\end{theorem}
\begin{proof}
Let $G$ and $H$ be two distinct labeled $EPR$ graphs defined
on the same set $V$ of vertices, such that both the graphs have
the same number of edges. For the sake of contradiction,
suppose $G$ and
$H$ are not LOCC incomparable. Then, by the definition of
LOCC incomparability,
either $G\geq H$ or $H\geq G$. Without loss
of generality, it can be assumed that $G\geq H$ (i.e.,
$G$ can be transformed to $H$ using LOCC).
Since both graphs have the same number $E(G)=E(H)$ of edges,
\begin{equation}\label{equsamenumedges}
\sum_{v\in V}deg_G(v)=2|E(G)|=2|E(H)|=\sum_{v\in V}deg_H(v)
\end{equation}
where $deg_G(v)$ ($deg_H(v)$) is the degree of the vertex $v\in V$
in EPR graph $G$  ($H$).
Further, by Lemma \ref{degreefall}, the degree of no vertex can
increase under LOCC.
Therefore, the degrees of all vertices remain unchanged.
Since the two graphs $H$ and $G$ are distinct,
there exists an edge $\{u,v\}$ in $G$,
which is not present in
$H$.
Define a bipartition $(\{u,v\},V\setminus \{u,v\})$ of the graph $G$ by
coloring vertices
$u$ and $v$ with color $1$, and the rest of the vertices
with color $2$.
See Figure \ref{figure3}.
The number of edges across the cut in this partition
is
$deg_{G}(u)+deg_{G}(v)-2$.
Since $(u,v)$ is not present in $H$,
the same cut due to the same bipartition of the vertices
will have $deg_{H}(v) + deg_{H}(u)$ edges in $H$.
Since the degree of each labeled vertex is the same
in both $G$ and $H$, the number of edges
in the reduced graph
after bicolored merging increases by 2.
This is not possible under LOCC
by Observation \ref{obscut}.
So, contrary to our assumption, the
two labeled EPR graphs $G$ and $H$ must be
identical.

\end{proof}

\begin{figure}[ht]
\rotatebox{0}{\centerline{\epsfig{figure=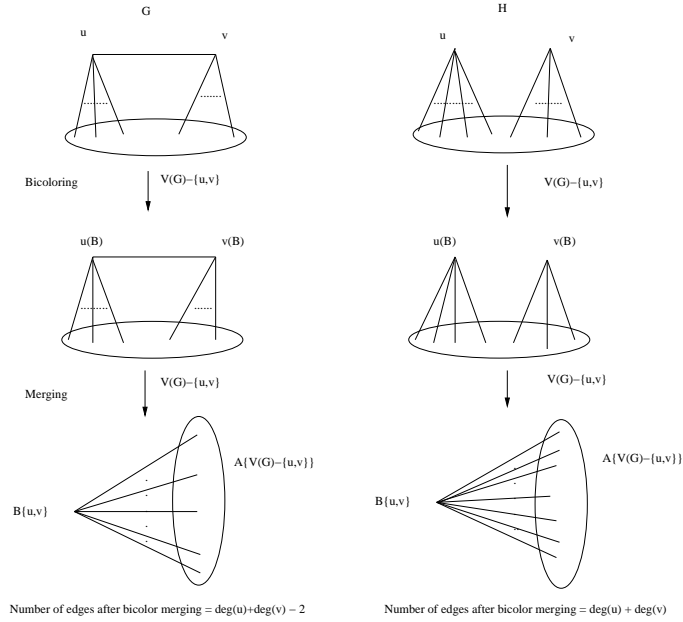,width=90mm}}}
\caption{{\it LOCC incomparability of two EPR graphs with equal number
of vertices and edges.}}
\label{figure3}
\end{figure}
$EPR$ graphs can be viewed as
$2$-uniform EC hypergraphs. It is natural to expect that
results shown above for EPR graphs hold also for
$r$-uniform EC hypergraphs for $r\geq 3$.

\begin{lemma}\label{l3}

Let $H_1$ and $H_2$ be
two {\it $r$-uniform EC hypergraphs}
defined on
the same vertex set $V$. If $H_1$ and $H_2$ have the same
number of hyperedges and $H_1\geq H_2$, then 
the degrees of all vertices in $H_{1}$ and $H_{2}$ are the 
same.

\end{lemma}
\begin{proof}
The degree of a vertex can not increase under $LOCC$ (see
Lemma \ref{degreefall}).
Also, $H_1$ and $H_2$ have same number of hyperedges. Therefore
the sum of degrees of all vertices in $H_1$ is equal to the sum of
the degrees of all vertices in $H_2$.
This enforces the degrees of all vertices to be same
for hypergraphs $H_1$ to $H_2$.
\end{proof}

In order to generalize the result as in Theorem \ref{EPRincom}
to $r$-uniform EC
hypergraphs, we
define the {\it degree of a vertex set}, generalizing the notion of
degree of a vertex.
\begin{definition}
For an EC hypergraph $H$ with vertex set $V$ and
subset $S\subseteq V$,
$deg_{H}(S)$ is defined as the number of hyperedges in
$H$ containing all the vertices of $S$.
\end{definition}

Observe that
$deg_{H}(S)$
is the (usual) degree of the vertex $v$
in $H$,
when $S=\{v\}$. The following lemma 
helps us in determining partial entropies, or equivalently, capacities of
cuts across bipartitions.
\begin{lemma}
For a subset $S$ of a $r$-uniform hypergraph $H$, the number of
hyperedges across the cut $(S, V\setminus S)$ is given by
\[
\sum_{F \subseteq S} (-1)^{|F|-1}deg_{H}(F)-\sum_{F \subseteq S, |F|=r} deg_{H}(F)
\]
\label{countcut}
\end{lemma}
\begin{proof}
Let $E$ be a hyperedge intersecting $S$ in $t\leq r$ vertices.
$E$ contributes to the first part of the sum above
through the terms  $(-1)^{|F|-1} deg_{H_i}(F)$, 
$\forall F \subseteq E \cap S$, by virtue of the {\it 
inclusion-exclusion principle}.
The contribution equals
(i) $1$, for the $t$ singleton
subsets $F \subseteq E \cap S$ with $|F|=1$, (ii) $-1$, for the
$\binom{t}{2}$ subsets $F\subseteq E \cap S$
with $|F|=2$, and so on, ending with $(-1)^{t-1}$, for the subset
$F=E \cap S$.
The total contribution of $E$ to the first part of the sum 
is $\displaystyle \sum_{i=1}^t (-1)^{i-1}\binom{t}{i}=1$.
The second part of the sum counts the number of hyperedges having all
vertices in $S$.
Hyperedge $E$ belongs to the cut  $(S, V\setminus S)$ if and only if
$0<t<r$. In this case $E$ contributes $+1$ to the first 
part and $0$ to the
second part of the sum, making a net contribution of $1$. 
For $t=r$ (and $t=0$)
the contribution of $E$ to both parts of the above sum
is 1 (and 0) respectively, making a net contribution of $0$.
Therefore, $\displaystyle \sum_{F \subseteq S} (-1)^{|F|-1}deg_{H}
(F)-\displaystyle\sum_{F \subseteq S, |F|=r} deg_{H}(F)$
equals the number of hyperedges across the cut  $(S, V\setminus S)$.
\end{proof}

We now proceed with the proof of the main LOCC incomparability
result for $r$-uniform EC hypergraphs using partial entropic 
criteria, for even integers $r\geq 4$.
Later, we present an example of two LOCC incomparable 
$3$-uniform EC hypergraphs
with equal number of hyperedges, where the partial entropic criteria
do not help in deciding incomparability. 

\begin{theorem}
\label{T1}
Let $H_1$ and $H_2$ be any two labeled {\it
$r$-uniform  EC hypergraphs} defined on the set $V$ of 
vertices. If $H_1$ and $H_2$ have the same number of hyperedges and
either $H_1\leq H_2$ or $H_2\leq H_1$
then (i) 
$deg_{H_1}(S)=deg_{H_2}(S)$, 
$\forall \displaystyle S \subseteq V$ 
such that $|S|< r$, 
for all integers $r\geq 3$, and
(ii) $H_1=H_2$, for all even integers $r\geq 4$.
\end{theorem}
\begin{proof}
We assume without loss of generality 
that $H_1$ is LOCC transformable to $H_2$.
It is sufficient to establish the following two claims.

\noindent \textbf{Claim (i):}
$\forall \displaystyle S \subseteq V$ 
such that $|S|< r$, $deg_{H_1}(S)=deg_{H_2}(S)$, for all 
integers $r\geq 3$.

\noindent \textbf{Claim (ii):}
$\forall \displaystyle S \subseteq V$ 
such that $|S|= r$, $deg_{H_1}(S)=deg_{H_2}(S)$, only for 
even integers $r\geq 4$.
[For sets $S \subseteq V$ with $|S|=r$,
$deg_{H_i}(S)=1$ if there is a hyperedge in $H_i$ consisting of
vertices in $S$, and $deg_{H_i}(S)=0$, otherwise.
Therefore, establishing the equality of $deg_{H_1}(S)$
and $deg_{H_2}(S)$ for
all subsets $S\subseteq V$ with $|S|=r$ implies
that the two hypergraphs are identical.]
Claim (i) is established by induction on the cardinality $k$ of
the $S\subseteq V$; Claim (ii) is subsequently established for 
even integers $r\geq 4$.

\noindent {\bf Proof of Claim (i):}

\noindent \textbf{Basis step:}
For $k=1$ the claim holds by Lemma \ref{l3}; if two $r$-uniform
EC hypergraphs with the same number of hyperedges
and vertices are not incomparable then they have the same vertex degrees.
In other words,
$deg_{H_1}(S)=deg_{H_2}(S)$, where
$S=\{v\}$, $\forall v \in V$.

\noindent \textbf{Induction hypothesis:}
Assume that
$deg_{H_1}(S)=deg_{H_2}(S)$,
for all $S\subseteq V$ such that
$|S|=k$, for all $k\leq m<r-1$.

\noindent \textbf{Induction Step:}
We require to show that 
$deg_{H_1}(S)=deg_{H_2}(S)$, 
for all $S\subseteq V$ such that $|S|=m+1$.


Clearly, no hyperedge in either hypergraph
can have all its vertices in $S\subseteq V$ as $|S|=k\leq m+1<r$. 
Therefore by Lemma \ref{countcut},
the number of edges
across the cut $(S,V\setminus S)$ in
$H_i, i=\{1,2\}$, is given by:
\begin{eqnarray}
&&\sum_{F \subseteq S} (-1)^{|F|-1}deg_{H_i}(F)
\label{normalcount}
\end{eqnarray}
Since the cut capacity cannot increase under LOCC
it follows that
\begin{eqnarray}\label{E}
&&\sum_{F \subseteq S} (-1)^{|F|-1}deg_{H_1}(F) \geq \sum_{F
\subseteq S} (-1)^{|F|-1}deg_{H_2}(F)\nn\\
\end{eqnarray}
By the induction hypothesis, 
$deg_{H_1}(F)=deg_{H_2}(F)$, for all $F\subseteq V$ with $|F|\leq m$.
Canceling out subsets which
contribute equally to both sides of
equation (\ref{E}) we have
\begin{eqnarray}\label{e1}
(-1)^{|S|-1}deg_{H_1}(S)&\geq&(-1)^{|S|-1}deg_{H_2}(S).
\end{eqnarray}
for all  $S \subset V$ with $|S|=m+1$.
Since the hypergraphs are $r$-uniform,
\begin{eqnarray}\label{e2}
\sum_{S\subseteq V,|S|=(m+1)}
deg_{H_i}(S)=\binom{r}{m+1} |E(H_{i})|, i\in \{1,2\}
\end{eqnarray}
where $E(H_1)=E(H_2)$ is the number of hyperedges in each EC hypergraph.
Given a hyperedge, ${r\choose m+1}$ subsets of the vertices
in that hyperedge would contribute 1 to the sum of the left hand side.
So, the total contribution over all hyperedges is the number of hyperedges in
$H_i$ times ${r\choose m+1}$.
Summing up relation (\ref{e1}) over all $S\subseteq V$,
\begin{eqnarray}
\displaystyle \sum_{S\subseteq V,|S|=(m+1)} (-1)^m deg_{H_1}(S)\geq \nonumber
\displaystyle \sum_{S\subseteq V,|S|=(m+1)}(-1)^m deg_{H_2}(S)
\end{eqnarray}
By equation (\ref{e2}) both the sums are equal
and therefore equality holds in the relation
(\ref{e1}), for all $S\subseteq V$ with $|S|=m+1$.
Therefore,
\begin{eqnarray}
deg_{H_1}(S)=deg_{H_2}(S)
\end{eqnarray}
for all $S\subseteq V$ with $|S|=m+1<r$.
Claim (i) of this theorem therefore holds by induction.


\noindent {\bf Proof of Claim (ii):}

Here we have $S\subseteq V$, such that $|S|=r$.
By Lemma \ref{countcut}, the number of 
edges across the cut $(S, V \setminus S)$
is given by:
\begin{eqnarray}
&&(\sum_{F \subseteq S} (-1)^{|F|-1}deg_{H_i}(F))-deg_{H_i}(S)
\label{evencount}
\end{eqnarray}
Since the cut capacity
does not increase under LOCC, like inequality (\ref{E}), 
the following inequality holds
\begin{eqnarray}
\left(\sum_{F \subseteq S} (-1)^{|F|-1}deg_{H_1}(F)\right)-deg_{H_1}(S) 
\geq 
\left(\sum_{F \subseteq S}
(-1)^{|F|-1}deg_{H_2}(F)\right)-deg_{H_2}(S) \label{E'}
\end{eqnarray}
By virtue of the already established Claim (i) above, 
$deg_{H_1}(F)=deg_{H_2}(F)$, for all $F\subseteq V$ with $|F|\leq m\leq r-1$.
Canceling out subsets which contribute equally to both sides of
equation (\ref{E'}), we get 
\begin{eqnarray}\label{E1}
2(-1)^{|S|-1}deg_{H_1}(S)\geq 2(-1)^{|S|-1}deg_{H_2}(S)
\end{eqnarray}
for even integers $r$.
The multiple 2 appears since the last (negative) term
$(-1)^{r-1}deg_{H_i}(S)$ in the summation 
on each side of the inequality 
\ref{E'} adds up with the negative term
$-deg_{H_i}(S)$,
for even integers $r\geq 4$.
Using remaining arguments as in the proof of Claim (i), 
and the equation (\ref{E1}), we have
\begin{eqnarray}
deg_{H_1}(S)=deg_{H_2}(S)
\end{eqnarray}
for all $S\subseteq V$, $|S|=r$, only for even $r\geq 4$, thereby
establishing Claim (ii).

\end{proof}

For an $r$-uniform hypergraph $H$ with $r$ odd,
the hyperedges having all vertices in $S$ contribute
$(-1)^{r-1}=1$ to the first term and $(-1)$ to the second term
in Lemma \ref{countcut}, and therefore cancel out. Therefore,
the number of hyperedges across a cut $(S,V\setminus S)$ is given by
\[
\sum_{F \subseteq S, |F|<r} (-1)^{|F|-1}deg_{H}(F)
\]
For all cuts, cut capacities are determined by
the quantities $deg_{H}(F)$ for $F\subset V$, $|F|<r$. We now exhibit
two 3-uniform hypergraphs on 6 vertices having 4 edges such that all cut
capacities are same in both the hypergraphs.
\[
H_{1}=\{123\},\{156\},\{245\},\{346\}
\]
\[
H_{2}=\{456\},\{234\},\{136\},\{125\}
\]

It is easy to verify that $deg_{H_{1}}(F)=deg_{H_{2}}(F)$ for $F\subset V$
and $|F|< 3$. From the above argument it follows that $H_{1}$ and
$H_{2}$ cannot be shown to be incomparable by 
partial entropic characterizations. 

As $deg_{H_{1}}(F)=deg_{H_{2}}(F)$ 
for $F\subset V$ and $|F|< 3$ so they are isentropic \cite{bennett00}. 
And from \cite{bennett00} that isentropic states are 
either LU (locally unitary) equivalent or incomparable.
Partition the
vertices into three sets $A=\{1\}, B=\{2, 3\}, C=\{4, 5, 6\}$ and 
merge the vertices in the same sets. $H_{1}$ reduces to
EPR graph with edges $(A, B), (A, C)$ and two copies of $(B, C)$. 
$H_2$ reduces to EC graph with 2 GHZ shared between $A, B, C$
and an EPR pair shared between $B$ and $C$. Let the reduced 
graph of $H_1$ and $H_2$ be denoted by $R(H_1)$ and $R(H_2)$
respectively. If $H_1$ and $H_2$ are LU equivalent then 
so is $R(H_1)$ and $R(H_2)$. To prove that 
they are not LU equivalent, observe that the mixed 
state obtained by tracing out $B$ from $R(H_2)$ state
i.e , $\rho_{AC}(R(H_{2}))$ is a maximally mixed, 
separable state of $A$ and $C$, while the 
corresponding mixed state $\rho_{AC}(R(H_1))$ obtained from   
$R(H_{1})$ can be distilled to entangled state, consisting 
on intact $(A, C)$ EPR pair shared by the two parties. So
if $R(H_1)$ and $R(H_2)$ are LU equivalent, then $A$ and $C$
can do local unitary operations and convert $\rho_{AC}(R(H_{2}))$
to $\rho_{AC}(R(H_{1}))$. This is not possible as one cannot 
make entanglement by $LOCC$, $R(H_{2})$ and $R(H_{2})$ are not
LU equivalent then $H_1$ and $H_2$ are not LU equivalent, 
therefore by \cite{bennett00} they are LOCC incomparable.

Next we propose an alternative proof of the incomparability result 
of \cite{sin:pal:kum:sri} about 
distinct $r$-uniform EC hypertrees as follows. 

\begin{theorem}
Any two distinct ensembles of multipartite CAT states 
represented by $r$-uniform EC hypertrees are
LOCC incomparable.
\end{theorem}
\begin{proof}

Let the two hypertrees defined on the vertex set 
$\{1,2,\cdots,n\}$ be $T_{1}$ and $T_{2}$.
For $r=2$ the result follows from Theorem \ref{EPRincom}
as all trees on $n$ vertices have $n-1$ edges.

For $r=3$, we assume without loss of generality that the hyperedge
$\{1,2,3\}$ is in $T_{1}\setminus T_{2}$. If $T_{1}$ and 
$T_{2}$ are not LOCC incomparable, then by part (i) in Theorem
\ref{T1}, we have $deg_{T_{1}}(\{1,2\})
=deg_{T_{2}}(\{1,2\})$. So, there should be a hyperedge 
$E1=\{1,2,x\}$ in $T_{2}\setminus T_{1}$, where $x$ is not in
$\{1,2,3\}$; this hyperedge $E_1$ cannot be in $T_1$ because no 
hypertree has two hyperedges with two common vertices. 
Similarly, 
$T_2$ must have hyperedges $E2=\{1,3,y\}$ and $E3=\{2,3,z\}$, where 
$y$ is not in $\{1,2,3,x\}$ and $z$ is not in $\{1,2,3,x,y\}$.  This 
implies that the cycle $\{3,E2,1,E1,2,E3,3\}$ is present in $T_{2}$,
a contradiction.

For $r>3$, we assume without loss of generality that the hyperedge
$\{1,2,\dots r\}$ is in $T_{1}\setminus T_{2}$.  
If $T_{1}$ and 
$T_{2}$ are not LOCC incomparable, then by part (i) of 
Theorem \ref{T1}, we have $deg_{T_{1}}(\{1,2,\dots,r-1\})
=deg_{T_{2}}(\{1,2,\dots,r-1\})$ and  $deg_{T_{1}}(\{2,3,\dots,r\})
=deg_{T_{2}}(\{2,3,\dots,r\})$. So, there must be hyperedges in
$T_{2}$ containing $\{1,2,\dots,r-1\}$ and 
$\{2,3,\dots,r\}$. As $r>3$, these two hyperedges 
intersect in at least 2 vertices, a contradiction.
\end{proof}

\section{Partial order induced by LOCC and its width}
\label{widthsec}

In this section we define a {\it partial order} where each 
node represents an equivalence class of states that are mutually LU
equivalent. This partial order is called the {\it LOCC partial 
order} and is represented by a {\it directed acyclic graph} $G_{LOCC}$.
We define this directed acyclic graph {\it $G_{LOCC}(V,E)$} as follows. 
Each vertex or node $v\in V$
represents an equivalence class of states that are mutually 
LU equivalent. The directed edge $(v,w)\in E$ (directed from
$v$ to $w$) exists in $G_{LOCC}(V,E)$
if any state in the
equivalence class of $v\in V$ can be transformed by LOCC to
a state in the equivalence class $w\in V$.
[By this definition of a directed edge, there is a
self loop in every node.]
We denote this partial order by the (binary) 
relation $\geq_{LOCC}$ between the nodes of the directed
graph $G_{LOCC}$;
for $X,Y\in V$, $X \geq_{LOCC} Y$, if and only if there 
is a directed edge from X to Y in
$G_{LOCC}$.
\begin{lemma}
The directed graph $G_{LOCC}(V,E)$ representing the LOCC partial order
is a {\it transitive graph} i.e., if
there is a directed edge $(v,w)\in E$, and a directed 
edge $(w,z)\in E$,
then there is also a directed 
edge $(v,z)\in E$.
\label{trans}
\end{lemma}
\begin{proof}
An edge from $v$ to $w$ implies that there exists an LOCC protocol
transforming a state $s\in v$ to an state $t\in w$. The directed edge
$(w,z)$ implies that
exists another LOCC protocol which converts the state $t\in w$ to a
state $u\in z$.
Applying the two protocols in succession,
$s$ can be converted to $u$, enforcing the 
directed edge $(w,z)$.
\end{proof}

\begin{corollary}
The graph $G_{LOCC}(V,E)$ has no non-trivial cycles.
\end{corollary}


It can be shown that multipartite quantum states form a partial order
under LOCC transformations (see \cite{pal:kum:sri}). 
Further the LOCC equivalent classes of quantum states
also form a partial order
under the relation $\geq_{LOCC}$ as defined above.
The relation $\geq_{LOCC}$ is a partial order as it satisfies the
three properties:
\begin{enumerate}
\item
The relation is reflexive since for all nodes of
$G_{LOCC}$ $X \geq_{LOCC} X$; each node
of $G_{LOCC}$ has a self loop.
\item
The relation is transitive, i.e., if $X \geq_{LOCC} Y$ and $Y \geq_{LOCC}
Z$
then $X \geq_{LOCC} Z$, as already shown earlier in Lemma \ref{trans}.
\item
The relation is antisymmetric. Since if $X \geq_{LOCC} Y$ and $Y
\geq_{LOCC} X$,
then $X$ is identical to $Y$ since there cannot be cycles
in $G_{LOCC}$ except for self
loops, as shown earlier.
\end{enumerate}
\begin{lemma}
The relation $\geq_{LOCC}$ among the nodes of the graph $G_{LOCC}$ forms 
a partial order.
\end{lemma}


The $width$ of $\geq_{LOCC}$ is the maximum number of
mutually LOCC incomparable
EC hypergraphs in $\geq_{LOCC}$. 
The width of $\geq_{LOCC}$ can be obtained using Theorem \ref{T1} and 
Sperner's Theorem \cite{sperner}.

\begin{Sperner}
The maximum cardinality of a collection of subsets of a 
$n$ element set, none of which contains another is 
$\binom{n}{\lfloor{n/2}\rfloor}$.
\end{Sperner}
 
Now we provide the derivation of the $width$ of the partial order
for $r$-uniform EC hypergraphs using Sperner's Theorem, 
where $r$ is an even integer.
\begin{theorem}
The maximum number of mutually LOCC incomparable $r$-uniform EC
hypergraphs (for even $r$), with $n$ nodes is $\binom{M}{\lfloor M/2
\rfloor}$, where $M=\binom{n}{r}$. 
\label{hygw}
\end{theorem}
\begin{proof}
Let $r>3$ be an even number. 
The maximum number of hyperedges possible in an $r$-uniform EC
hypergraph 
with $n$ vertices is $M$= $\binom{n}{r}$. 
Let $S$ be any set of $N> \binom{M}{\lfloor M/2 \rfloor}$
mutually LOCC incomparable distinct $r$-uniform EC hypergraphs 
defined on $n$ vertices.
By Sperner's Theorem, there must be two hypergraphs $H_1$ and $H_2$
in the collection
$S$ such that the set of hyperedges in $H_1$ is a subset of the
set of hyperedges in $H_2$. This contradicts the assumption that 
$H_1$ and $H_2$ are LOCC incomparable because $H_1\geq_{LOCC}H_2$ by
a simple LOCC transformation that drops all additional hyperedges 
in $H_1\setminus H_2$ from $H_1$.
So, we know that $N\leq \binom{M}{\lfloor M/2 \rfloor}$, is an upper
bound on the width of the partial order $\geq_{LOCC}$.

Now we show that this bound is tightly achievable as follows.
Consider the set of all the different
$r$-uniform EC hypergraphs on $n$ vertices with exactly a fixed number 
$\lfloor \frac M 2\rfloor$ of 
hyperedges. By Theorem \ref{T1}, all these EC hypergraphs 
are LOCC incomparable. 
This set has $\binom{M}{\lfloor M/2 \rfloor}$
$r$-uniform EC hypergraphs,
forming an {\it antichain} of the partial order $\geq_{LOCC}$. 
Therefore the width of the partial order $\geq_{LOCC}$
is $\binom{M}{\lfloor M/2 \rfloor}$ where
$M=\binom{n}{r}$.
\end{proof}
 
\section{Restriction of LOCC to teleportation and EPR destruction}

\label{cutssec}



Now we consider LOCC restricted to the two
basic operations of {\it edge destruction}
and {\it teleportation} on EPR graphs. The allowed
operations are (i) discarding an edge
(destroying an EPR pair) and
(ii) teleportation, i.e., replacing
edges (EPR pairs)  $\{x, y\}$ and $\{y, z\}$
by one EPR pair across
edge $\{x, z\}$.
If EPR graphs $G$ and $H$ are such that $H$ can be obtained from  
$G$ by such restricted LOCC, then we say $G\geq_R H$.
In the following lemma we characterize graph theoretic 
properties such that $G\geq_R H$, for EPR graphs $G$ and $H$.

\begin{lemma}
\label{disjointpaths}
Let $G$ and $H$ be
two EPR graphs defined on the same vertex set $V$. Then,
$G\geq_R H$ if and only if there
are edge disjoint paths in $G$ from $u$ to $v$ 
for all edges $\{u,v\}$ in $H$.
\end{lemma}
\begin{proof}
For the if part, observe that a
path $P$ from $u$ to $v$ in $G$ can be
reduced to the EPR edge \{$u$, $v$\} using
successive steps of teleportation.
\begin{eqnarray}
u v_1 v_2 v_3 \dots v_k v \rightarrow
u v_1 v_3 \dots v_k v \rightarrow 
\dots \rightarrow u v
\end{eqnarray}
Therefore edge disjoint paths from $u$ to $v$ can be
converted independently to edges 
\{$u$,$v$\} in $H$.
The remaining edges of $G$ on this path are
discarded to transform $G$ to $H$.

For the only if part,
consider the inverse of the two
possible operations
\begin{enumerate}
\item
Adding an edge
\item
Replacing an edge by a path of length $2$.
\end{enumerate}
If $G\geq_R H$, we can create $G$ from $H$ 
using the inverse operations as follows.
As we keep applying these inverse operation steps for constructing $G$
from $H$, the invariant maintained is the presence of  
at least $|E(H)|$ edge-disjoint
paths in the intermediate graphs. 
Initially we have one path
in $H$ from $u$ to $v$ for each edge \{$u$, $v$\} in $H$. 
These edges themselves are the initial edge disjoint paths 
to begin with.
Since the application of the first inverse operation does not destroy any 
edge, we need to consider only the second operation.
During the application of the second inverse operation, one edge
of a path may be destroyed; however, two edges are added to reconnect the
path thereby preserving the edge disjointness property of all the 
relevant paths.
Therefore $G$ contains
edge disjoint paths from $u$ to $v$, for each $\{u,v\}\in E(H)$.
\end{proof}

As shown earlier, partial entropic criteria are
not applicable for establishing $LOCC$ incomparability
of certain kinds of multipartite states.
The new criterion of the existence of
edge disjoint paths in the EPR graph $G$
for each edge in $H$, provides a stronger
characterization of proving $LOCC$ incomparability in
the restricted model.

\begin{figure}[ht]
\rotatebox{0}{\centerline{\epsfig{figure=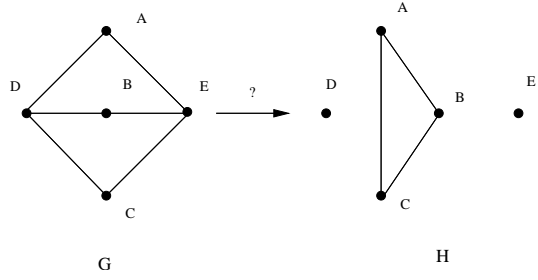,width=70mm}}}
\caption{Two EPR graphs incomparable under restricted LOCC, not
yielding to partial entropic criteria.} \label{figure1}
\end{figure}

Consider two EPR graphs $G$ and $H$ in Figure \ref{figure1}.
Since $G$ has more edges than $H$, $H$ cannot be transformed to $G$
by LOCC. So, in order to show that 
EPR graphs $G$ and $H$ are LOCC incomparable, we need to only show that 
$G\geq H$ does not hold.
Observe that using the bicolored merging technique cannot help us
establish that $G\geq H$ does not hold.  
This is due to the fact that no bicoloring of
the vertex set results in any violation of the non-increase
of partial entropy as we go from $G$ to $H$. 
Observe however that $G\geq_R H$ does not hold by Lemma \ref{disjointpaths}
since $H$ contains three new edges
$\{A,B\}$, $\{B,C\}$, $\{C,A\}$, but $G$ does not
contain edge disjoint paths from $A$ to $B$, $B$ to $C$, and $C$
to $D$. 
So, the two graphs are $LOCC$ incomparable under our restricted $LOCC$ model. 
A natural open question is whether the two EPR graphs $G$ and $H$
are LOCC incomparable in the general model. 
We conjecture that they are indeed LOCC incomparable.

We now investigate whether the edge disjoint path
criterion is powerful enough to capture LOCC. 
We show that this is indeed the case when 
edges in $E(H)\setminus E(G)$
appear one at a time.

\begin{definition}
An LOCC transformation from $G$ to $H$ is called good
if $|E(H)\setminus E(G)|\leq 1$, where $G$ and $H$ are EPR graphs defined 
on the same vertex set $V$.
\end{definition}

\begin{lemma}
\label{good}
Suppose EPR graph $G$ can be transformed to 
$H$ via a good transformation. Then, 
$G\geq_R H$.
\end{lemma}
\begin{proof}
If $E(H)\setminus E(G)$ is empty, we can create $H$ from $G$ by
a sequence of EPR edge destructions. 
For the case where $E(H)\setminus E(G)=1$, 
we present a constructive proof depicting  
a path from $u$ to $v$ in $G$, for the single edge 
$\{u,v\}\in E(H)\setminus E(G)$, where all the edges of the 
constructed path belong to $E(G)\setminus E(H)$.
This path is sufficient to establish $G\geq_R H$.

We construct the path from $u$ to $v$ in $E(G)\setminus E(H)$ 
as follows.
Initialize the vertex set $C_1=\{u\}$ and the set of edges
defined on $C_1$ as $E(C_1)=\phi$. Set  $i=1$.
Perform the following steps until termination in Step 4.
\begin{enumerate}
\item
Consider the cut 
$(C_i,V\setminus C_i)$. Since edge $\{u,v\}\in E(H)\setminus E(G)$,  
the capacity of this cut cannot 
increase under LOCC. So, there must be vertices 
$u_i^{'} \in C_i$ and $w_i\in V\setminus C_i$
such that the edge $\{u_i^{'},w_i\}$ is in $E(G)\setminus E(H)$.
Find such an edge $\{u_i^{'},w_i\}$.
\item $C_{i+1}=C_i\cup \{w_i\}$.
\item
$E(C_{i+1})=E(C_i)\cup {\{u_i^{'}, w_i\}}$ i.e., add the edge
to $E(C_{i})$, yielding tree $E(C_{i+1})$ over vertex set $C_{i+1}$.
\item
If $w_i=v$ then stop else $i=i+1$.

\end{enumerate}
The invariant at the beginning of each iteration of the 
above procedure is that the subgraph $G(C_i,E(C_i))$ of $G$ with
vertex set $C_i$ and edge set $E(C_i))$ is a tree; 
this subgraph is connected, and has exactly $|C_i|-1$ edges. 
When the process terminates,
the tree $E(C_i)$ in $E(G)\setminus E(H)$ contains
both $u$ and $v$. So, there is a path from $u$ to $v$ using
edges entirely from $E(G)\setminus E(H)$.
 
Using this path
we can perform LOCC transformations in our restricted model 
using repeated teleportation steps, thereby creating the 
only new EPR edge $\{u,v\}\in E(H)\setminus E(G)$. 
We can destroy the remaining 
EPR pairs from $G$ that do not
belong to $E(H)$, finally yielding $H$. So, $G\geq_R H$. 
\end{proof}

If an EPR graph $H$ can be obtained from an EPR graph 
$G$ by a sequence of good transformations, then 
we know that $G\geq_R H$ by the repeated application of 
Lemma \ref{good}. For the converse, suppose $G\geq_R H$. Then, by 
Lemma \ref{disjointpaths}, there are edge disjoint paths in $G$
for all edges $\{u,v\}$ in $H$. In order to generate $H$ from $G$, 
we may therefore use such disjoints paths, one at a time, to
generate the edges in $H$ that do not exist in $G$. Each such 
transformation is a good transformation since it generates 
at most one new EPR edge in the resulting intermediate EPR graph; 
using a sequence of such good transformations, $G$ can be converted to $H$.
We now summarize our characterizations as follows.

\begin{theorem}
Let $G$ and $H$ be EPR graphs defined on the same vertex set. The
following statements are equivalent.
\begin{enumerate}
\item $G\geq_R H$.
\item There are edge disjoint paths in $G$ from $u$ to $v$ for
all edges $\{u,v\}$ in $H$.
\item $H$ can be obtained from $G$ by a sequence of good transformations.
\end{enumerate}
\label{characthm}
\end{theorem}

Given two EPR graphs $G$ and $H$,
the decision problem of determining whether $G \geq_R H$
is NP-hard since this problem can be used to solve the decision problem 
of checking for edge disjoint paths in $G$. The problem of deciding the
existence of edge disjoint paths in graphs was shown to
be NP-complete in \cite{disjoint}. 

It remains open to determine whether our restricted LOCC $\geq_R$ 
is powerful enough to capture (general) LOCC for EPR graphs.
We believe that the two models are equally powerful for EPR 
graphs.

\section{Concluding remarks}
Partial entropic criteria are not sufficient for 
demonstrating LOCC incomparability between multipartite states. New techniques 
need to be developed for a better understanding
of LOCC comparability. Further, it would be interesting 
to investigate whether the restricted 
model of LOCC studied in this paper (which uses only teleportation and 
EPR pair destruction),
is powerful enough to capture LOCC in general for multipartite states
comprising multiple EPR pairs. 
For LOCC incomparable ensembles, the amount of
quantum communication necessary for transformations, and the possibility
of classifications based on some notions of quantum distance between ensembles
may be studied.

\vspace{.2cm}
\noindent {\it Acknowledgments:} The authors would like to thank Simone
Severini of the Institute for Quantum Computing, University of Waterloo
for discussions, and the referees of a previous version of this paper
for their helpful comments and suggestions. 
S. P. Pal acknowledges Sudhir Kumar Singh of UCLA
and R. Srikanth of Raman Research Institute, Bangalore 
for discussions.
Arijit Ghosh's research at Perimeter Institute for Theoretical Physics was 
supported in
part by the Government of Canada through NSERC and by the Province of
Ontario through MRI.

%
%
%
%

\appendix

\section{Bicolored merging and partial entropy}




In this appendix we show the equivalence of {\it partial
entropic criteria} and the technique of {\it bicolored merging}
in establishing LOCC incomparability of 
multipartite states represented by EC hypergraphs and EPR graphs.
Let $A$ and $B$ be disjoint sets of parties sharing the
quantum state $\rho^{AB}$ between them.
If $\rho^{AB}=\rho \otimes \sigma$, where $\rho$ is the
density operator of the system $A$, and $\sigma$ is a
density operator for the system $B$, 
then 
we know
from page 106. of
\cite{nc} 
that the partial entropy 
\begin{equation}
\rho^{A}=tr_{B}(\rho^{AB})=tr_{B}(\rho \otimes\sigma)=\rho. \label{E9}
\end{equation}
Also, if $\rho^{AB}=\rho$ where
$\rho$ is the density operator of the system $A$. Then the partial entropy
\begin{eqnarray}
\label{E200} \rho^{A}=tr_{B}(\rho^{AB})&=&tr_{B}(\rho)=\rho.
\end{eqnarray}
We also use the important property of {\it Von-Neumann entropy} from
page 514 of \cite{nc} that,
\begin{eqnarray}
\label{E20}
S(\rho\otimes\sigma)&=&S(\rho) + S(\sigma)
\end{eqnarray}
Let the parties of $A$ and $B$
share an $n$-CAT state where the first $r$ qubits is
with the parties of the set $A$ and the remaining
qubits from $r+1$ to $n$ is with the parties of the set
$B$. This state has the density operator
\begin{equation}
\rho^{AB}=\left(\frac{|0_{1}\dots 0_{r}0_{r+1}\dots0_{n}\rangle+
           |1_{1}\dots
           1_{r}1_{r+1}\dots1_{n}\rangle}{\sqrt{2}}\right)\nonumber
           \left(\frac{\langle0_{1}\dots 0_{r}0_{r+1}\dots0_{n}| +
           \langle1_{1}\dots 1_{r}1_{r+1}\dots1_{n}|}{\sqrt{2}}\right)
\end{equation}
We denote  $|0_{1}\dots 0_{r}\rangle$ by $|0_{A}\rangle$, and
$|1_{1}\dots 1_{r}\rangle$ by $|1_{A}\rangle $. We use similar notation for
the set $B$. Then,
\begin{eqnarray}
\label{E3}
\rho^{AB}&=&\left(\frac{|0_{A}0_{B}\rangle+|1_{A}1_{B}\rangle}{\sqrt{2}}\right)
     \left(\frac{\langle0_{A}0_{B}| + \langle1_{A}1_{B}|}{\sqrt{2}}\right)\nonumber\\
         &=&\frac{|0_{A}0_{B}\rangle\langle0_{A}0_{B}|+
            |1_{A}1_{B}\rangle\langle0_{A}0_{B}|}{2}
            +\frac{|0_{A}0_{B}\rangle\langle1_{A}1_{B}|+
            |1_{A}1_{B}\rangle\langle1_{A}1_{B}|}{2}
\end{eqnarray}
Tracing out the system $B$ from $\rho_{AB}$, we find
the reduced density operator of the system $A$,
\begin{eqnarray}
\label{E4}
\rho^{A}&=&tr_{B}(\rho^{AB})\nonumber\\
        &=&\frac{tr_{B}|0_{A}0_{B}\rangle\langle0_{A}0_{B}|+
            tr_{B}|1_{A}1_{B}\rangle\langle0_{A}0_{B}|}{2}
            +\frac{tr_{B}|0_{A}0_{B}\rangle\langle1_{A}1_{B}|+
            tr_{B}|1_{A}1_{B}\rangle\langle1_{A}1_{B}|}{2}\nonumber\\
        &=&\frac{|0_{A}\rangle\langle0_{A}|\langle0_{B}|0_{B}\rangle +
                 |1_{A}\rangle\langle0_{A}|\langle0_{B}|1_{B}\rangle}{2}
                 +\frac{|0_{A}\rangle\langle1_{A}|\langle1_{B}|0_{B}\rangle +
                 |1_{A}\rangle\langle1_{A}|\langle1_{B}|1_{B}\rangle}{2}\nonumber\\
        &=&\frac{|0_{A}\rangle\langle0_{A}|\langle0_{B}|0_{B}\rangle+
                 |1_{A}\rangle\langle1_{A}|\langle1_{B}|1_{B}\rangle}{2}\nonumber\\
        &=&\frac{|0_{A}\rangle\langle0_{A}|+|1_{A}\rangle\langle1_{A}|}{2}\nonumber\\
\end{eqnarray}
From \ref{E4} we get,
\begin{eqnarray}
\label{E10}
S_{B}(\rho^{AB})&=&-tr\left(\rho^{A}\log_{2}\rho^{A}\right)\nonumber\\
               &=&1.
\end{eqnarray}
So, from \ref{E10} we conclude that $S_{B}(\rho^{AB})=1$ if an 
$n-CAT$ is shared by parties of both the sets. Also, from \ref{E9}
and \ref{E200}, we have
\begin{eqnarray}
S_{B}(\rho^{AB})=0
\label{EXX}
\end{eqnarray}
if all entanglements are shared 
only by parties within sets
$A$ and $B$ but not across $A$ and $B$. 

We proceed to prove our Theorem \ref{bcm}.
Let $H$ denote and EC hypergraph. Let $V(H)=C\bigcup D$,
where $C$ and $D$ are disjoint.
We use the following notation.
Let $D(H)$ (or $C(H)$) denote the set of 
hyperedges shared within the elements of
the set $D$ (or $C$).
Let $CD(H)$ denote the set of hyperedges shared across the 
sets $C$ and $D$.
We say $H_1 \ngeq H_2$ if  
$H_1 \geq H_2$ does not hold.  

\begin{theorem}
\label{bcm}
Let $H_1$ and $H_2$ be two {\it entanglement hypergraphs}
shared between the geographically separated
parties $p_1, p_2, \dots, p_n$. If it can be shown
using partial entropic criteria
that $H_1 \ngeq H_2$, 
then there exists a bicolored merging scheme establishing
$H_1 \ngeq H_2$. 
\end{theorem}
\begin{proof}
Suppose there is a subset $X$ of ${p_1,p_2, \dots,p_n}$ such that
$S_{X}(H_1)<S_{X}(H_2)$, thereby ensuring that $H_1$ cannot be transformed
to $H_2$ using LOCC. Since each hyperedge corresponds to a $GHZ$ state
in the EC hypergraph, we denote the corresponding maximal
entanglement associated with the hyperedge $e$ as $|e\rangle$. 
Let 
the hyperedges of $H_1$ be $e_{11}, e_{12}, \dots,
e_{1r}, \dots$, and those belonging to $H_2$ be $e_{21}, e_{22},
\dots, e_{2j}, \dots$. For $H_1$,
\begin{equation}
\rho^{H_1}= \bigotimes_{e_{1i}\in E(H_1)}|e_{1i}\rangle \langle
e_{1i}|
\end{equation}
Therefore,
\begin{eqnarray}
\label{E11}
\rho_{X}^{H_1}&=& tr_{X}\bigotimes_{e_{1i} \in E(H_1)}|e_{1i}\rangle \langle e_{1i}| \nonumber\\
              &=&\bigotimes_{e_{1i} \in E(H_1)}tr_{X}|e_{1i}\rangle \langle e_{1i}| \\
              &=&\left(\bigotimes_{e_{1s}\in X(H_1)}tr_{X}|e_{1s}\rangle \langle
              e_{1s}|\right)
    \otimes\left(\bigotimes_{e_{1t}\in \bar X(H_1)}tr_{X}|e_{1t}\rangle \langle
   e_{1t}|\right)\nn
    \otimes\left(\bigotimes_{e_{1u}\in X\bar X(H_1)}tr_{X}|e_{1u}\rangle \langle
    e_{1t}|\right)
\end{eqnarray}
From $\ref{E9}$ we get,
\begin{equation}\label{E12}
\rho_{X}^{H_1}=\left(\bigotimes_{e_{1t}\in \bar
                    X(H_1)}tr_{X}|e_{1t}\rangle \langle
                    e_{1t}|\right)
                    \otimes
                    \left(\bigotimes_{e_{1u}\in X\bar X(H_1)}tr_{X}|e_{1u}\rangle \langle
                    e_{1u}|\right)
\end{equation}
From \ref{E200} we get,
\begin{equation}
\label{E13} \rho_{X}^{H_1}=\left(\bigotimes_{e_{1t}\in \bar
X(H_1)}|e_{1t}\rangle \langle e_{1t}|\right)
    \otimes\left(\bigotimes_{e_{1u}\in X\bar X(H_1)}tr_{X}|e_{1u}\rangle \langle e_{1u}|\right)
\end{equation}
We know that
$S_{\bar X}(H_1)= S(\rho_{X}^{H_1})$.
From equations \ref{E20} and \ref{EXX} we get,
\begin{eqnarray}
S_{\bar X}(H_1)&=&S\left(\bigotimes_{e_{1t}\in \bar
                  X(H_1)}|e_{1t}\rangle \langle e_{1t}|\right)
                  + S\left(\bigotimes_{e_{1u}\in X\bar X(H_1)}tr_{X}|e_{1u}\rangle
                  \langle e_{1u}|\right)\nn\\
               &=&S\left(\bigotimes_{e_{1u}\in X\bar
                  X(H_1)}tr_{X}|e_{1u}\rangle \langle e_{1u}|\right)
\end{eqnarray}
From \ref{E20} we get,
\begin{equation}
S_{\bar X}(H_1)=\sum_{e_{1u}\in X\bar
X(H_1)}S\left(tr_{X}|e_{1u}\rangle \langle e_{1u}|\right) \nn
\label{E40} 
\end{equation}
We know that,
\begin{equation}
S\left(tr_{X}|e_{1u}\rangle \langle
e_{1u}|\right)=1,~~\forall~~|e_{1u}\rangle~~\in~~X\bar X(H_1)
\label{E50} 
\end{equation}
From \ref{E50},
\begin{eqnarray}
S_{\bar X}(H_1)&=&\mbox{number of hyperedges containing
at
least one}\nn \\
&&\mbox{party of $X$ and $\bar X=\{p_1,p_2, \dots,p_n\}\setminus X$}\nn\\
\label{E5} 
\end{eqnarray}
Similarly for $H_2$,
\begin{eqnarray}
S_{\bar X}(H_2)&=& S(\rho_{X}^{H_2}) \nonumber\\
               &=&-tr{\rho_{X}^{H_2}\log_2 \rho_{X}^{H_2}}\nonumber\\
               &=& \mbox{number of hyperedges containing at least one}\nonumber\\
                && \mbox{party of $X$ and $\bar X=\{p_1,p_2, \dots,p_n\}\setminus X$}\nonumber\\
\label{E5111} 
\end{eqnarray}

Now we use the bicolored merging scheme. We color all the
vertices in the set $X$ with one color and collapse them into 
the merged party $A$; we color the rest of the vertices with 
another color and collapse them into another merged party 
$B$. The number of (hypothetical) EPR pairs in the reduced bicolor-merged
graph (BCM graph) $H_{1}^{bcm}$ 
(or $H_{2}^{bcm}$) 
after bicolored merging, is equal to
number of hyperedges containing elements of both $X$ and
$\{p_1,p_2, \dots,p_n\} \setminus X$, which is equal to $S_{X}(H_1)$
(or $S_{X}(H_2)$)
from equation \ref{E5} (\ref{E5111}). Since
$S_{X}(H_1) < S_{X}(H_2)$, the number of
$EPR$ pairs in $H_{1}^{bcm}$ is less than those 
present in $H_{2}^{bcm}$. So, $H_{1}$ cannot be
transformed to $H_{2}$ using LOCC.
\end{proof}

The number of edges in the reduced graph obtained after bicolored
merging is equal to the capacity of the corresponding cut.
The principle of bicolored merging can therefore be restated as follows:
\begin{lemma}\label{cut}
Suppose $G$ and $H$ are EC hypergraphs such that $G \geq H$.
Then the cut capacity across a cut in $H$ cannot be greater
than the cut capacity across the same cut in $G$.
\end{lemma}


\end{document}